\documentclass[a4paper,11pt,twoside,notitlepage,reqno]{amsart}

\usepackage{amsmath,amssymb,amsbsy,amsthm}
\usepackage[hmargin=1.4in,vmargin=1.4in]{geometry}
\usepackage{url}

\usepackage{tikz}
\usetikzlibrary{matrix}
\usepackage[linktocpage]{hyperref}
\hypersetup{
    colorlinks=true,        
    linkcolor=red,          
    citecolor=red,         
    filecolor=magenta,      
    urlcolor=cyan           
}

\addtocontents{toc}{\protect\setcounter{tocdepth}{1}}

\usepackage{eucal}

\newcommand{\Z}{ \mathbb{Z} }
\newcommand{\N}{ \mathbb{N} }

\DeclareMathOperator{\theory}{FO}

\DeclareMathOperator{\factoreq}{factoreq}
\DeclareMathOperator{\shift}{shift} \DeclareMathOperator{\conj}{conj} 
\DeclareMathOperator{\lessthan}{lessthan}
\DeclareMathOperator{\lessthaneq}{lessthaneq}
\DeclareMathOperator{\allconj}{allconj}
\DeclareMathOperator{\lexleast}{lexleast}
\DeclareMathOperator{\lie}{lie}

\theoremstyle{plain}
	\newtheorem{theorem}{Theorem}
	\newtheorem{notation}[theorem]{Notation}
		\numberwithin{theorem}{section}
	\newtheorem{lemma}[theorem]{Lemma}
	\newtheorem{proposition}[theorem]{Proposition}
	\newtheorem{corollary}[theorem]{Corollary}
	
\newtheorem{question}[theorem]{Question}
	\newtheorem*{theorem*}{Theorem}
	\newtheorem*{lemma*}{Lemma}
	\newtheorem*{prop*}{Proposition}
	\newtheorem*{cor*}{Corollary}
	\newtheorem*{conj*}{Conjecture}

\theoremstyle{definition}
	\newtheorem{example}[theorem]{Example}
	\newtheorem*{example*}{Example}
	
	\newtheorem{remark}[theorem]{Remark}

\begin{document}


\title{Lie complexity of words}

\author{Jason P. Bell}
\address{Department of Pure Mathematics\\
University of Waterloo\\
Waterloo, ON N2L 3G1\\
Canada}
\email{jpbell@uwaterloo.ca}
\thanks{Jason Bell is supported by NSERC grant 2016-03632. Jeffrey Shallit is supported by NSERC grant 2018-04118.}

\author{Jeffrey Shallit}
\address{School of Computer Science\\
University of Waterloo\\
Waterloo, ON N2L 3G1\\
Canada}
\email{shallit@uwaterloo.ca}

\begin{abstract} 
Given a finite alphabet $\Sigma$ and a right-infinite word $\bf w$ over $\Sigma$,  we define the Lie  complexity function $L_{\bf w}:\mathbb{N}\to \mathbb{N}$, whose value at $n$ is the number of conjugacy classes (under cyclic shift) of length-$n$ factors $x$ of $\bf w$ with the property that every element of the conjugacy class appears in $\bf w$. 

We show that the Lie complexity function is uniformly bounded for words with linear factor complexity, and as a result we show that words of linear factor complexity have at most finitely many primitive factors $y$ with the property that $y^n$ is again a factor for every $n$.

We then look at automatic sequences and show that the Lie complexity function of a $k$-automatic sequence is again $k$-automatic. \end{abstract}
\keywords{Combinatorics on words, automatic sequences, morphic words, linear factor complexity, Lie complexity.}
\subjclass[2020]{68R15, 11B85}

\maketitle

\tableofcontents

\section{Introduction}
Let $\Sigma$ be a finite alphabet and let ${\bf w}$ be a right-infinite word over $\Sigma$.  The {\it factor complexity function\/} $p_{\bf w}:\mathbb{N}\to \mathbb{N}$, which counts the number of factors of ${\bf w}$ of each length, plays a fundamental role in understanding the behaviour of ${\bf w}$ as a word; see, e.g., \cite{Cassaigne:1996,Mignosi:1989}.   (It is also called the {\it subword complexity function}.)   Often, however, one wishes to understand factors of ${\bf w}$ of a special form (e.g., palindromes, bordered, unbordered, squarefree, repetition-free, $k$-power, etc.).  To accomplish this task, one requires the use of finer invariants, which are designed to count factors of a certain form.  Of course, it is generally a very difficult problem to exactly count factors of a specific form in a given word, and so in practice one settles for invariants that are easier to compute, which give upper and lower bounds for the desired quantities.

In this paper, we look at the problem of counting factors $y$ of a right-infinite word ${\bf w}$ with the property that all cyclic shifts of $y$ remain factors of ${\bf w}$.  In particular, this includes factors $y$ with \emph{unbounded exponent} (that is, factors $y$ of ${\bf w}$ with the property that $y^n$ is again a factor of ${\bf w}$ for every $n\ge 1$), and so a consequence of our work is that we are able to give upper bounds on the number of such factors. In the case of counting factors $y$ of unbounded exponent, we may restrict our attention to the case when $y$ is itself not a perfect power; that is, when $y$ is \emph{primitive}. In this case, it is known that the set of primitive factors $y$ of ${\bf w}$ having unbounded exponent is a finite set when $\bf w$ is pure morphic (fixed point of a morphism).

Since automatic words have linear factor complexity (that is, the number of factors of length $n$ is bounded by a fixed affine function $An+B$ in $n$ for every $n$), it is natural to ask whether a similar phenomenon holds more generally for words of linear factor complexity.  To accomplish this, we introduce a new complexity function, the \emph{Lie complexity function}, which is motivated by ideas from the theory of Lie algebras.   Given a right-infinite word ${\bf w}$, we define its Lie complexity function, $L_{\bf w}:\mathbb{N}\to \mathbb{N}$, to be the map in which $L_{\bf w}(n)$ is equal to the number of equivalence classes $[y]$ of length-$n$ factors of ${\bf w}$ with the property that every cyclic permutation of $y$ is again a factor of ${\bf w}$.  Our main theorem is the following estimate.
\begin{theorem}
Let $\Sigma$ be a finite alphabet, let ${\bf w}$ be a right-infinite word over $\Sigma$, and let $L_{\bf w}:\mathbb{N}\to \mathbb{N}$ be the Lie complexity function of ${\bf w}$.  Then for each $n\ge 1$ we have
$$L_{\bf w}(n)\le p_{\bf w}(n)-p_{\bf w}(n-1)+1.$$  In particular, if ${\bf w}$ has linear factor complexity, then $L_{\bf w}(n)$ is uniformly bounded above by a constant.
\label{thm:main1}
\end{theorem}
Observe that if ${\bf w}$ is a right-infinite  word and $y$ is a primitive word such that $y^n$ is a factor of ${\bf w}$ for every $n$, then for every $n$, all cyclic permutations of $y^n$ are necessarily factors of ${\bf w}$.  Using this observation, we are able to prove the following result.
\begin{theorem}
Let $\Sigma$ be a finite alphabet and let ${\bf w}$ be a right-infinite word over $\Sigma$.  If ${\bf w}$ has linear factor complexity, then the set of primitive factors $y$ of ${\bf w}$ such that $y^n$ is a factor of ${\bf w}$ for every $n$ is a finite set.   
\label{thm:main2}
\end{theorem}
We point out that an analogue of Theorem \ref{thm:main2} was already known to hold for pure morphic sequences \cite[Corollary 20]{KS}.
We are also able to show that the condition that $\limsup_{n \rightarrow \infty} p_{\bf w}(n)/n$ be finite in Theorem \ref{thm:main2} cannot be relaxed.
\begin{theorem} Let $f:\mathbb{N}\to \mathbb{N}$ be a function that tends to infinity as $n\to \infty$ and let $\Sigma$ be a finite alphabet.  Then there is a right-infinite recurrent word ${\bf w}$ over $\Sigma$ such that $p_{\bf w}(n)\le n f(n)$ for $n$ sufficiently large such that ${\bf w}$ has infinitely many distinct primitive factors $y$ with the property that $y^n$ is a factor of ${\bf w}$ for every $n$.
\label{thm:main3}
\end{theorem}
We next turn our attention to automatic words ${\bf w}$.  If $k\ge 2$ is a positive integer, $\Delta$ is a finite set, and $f:\mathbb{N}\to \Delta$ is a $k$-automatic sequence, then we can identify $f$ with the right-infinite word ${\bf w}:=f(0)f(1)f(2)\cdots $ over the alphabet $\Delta$.  Thus it makes sense to talk about the Lie complexity function of the automatic sequence $f$, by making this identification with the word ${\bf w}$.  Our next result shows that the Lie complexity functions of automatic words are particularly well-behaved.

\begin{theorem}
Let $k\ge 2$ be a positive integer, let $\Delta$ be a finite set, and let $f:\mathbb{N}\to \Delta$ be a $k$-automatic sequence.  Then the Lie complexity function of $f$ is again a $k$-automatic sequence. 
\label{thm:main5}
\end{theorem}
The outline of this paper is as follows.  In \S\ref{sec:Lie} we formally introduce the Lie complexity function and in \S\ref{sec:alg} we give an algebraic interpretation of this complexity function.  In \S\ref{sec:proof} we use the algebraic theory we develop to prove Theorems \ref{thm:main1} and \ref{thm:main2}.  In \S\ref{sec:con}, we give a construction which proves Theorem \ref{thm:main3}.  In \S\ref{sec:aut} we prove Theorem \ref{thm:main5} and then give some examples in \S\ref{sec:exam}. Finally, in \S\ref{sec:Q} we raise a questions about whether the Lie complexity function of a morphic word is necessarily uniformly bounded.

\section{Lie complexity}\label{sec:Lie}
Let $\Sigma$ be a finite alphabet and ${\bf w}$ be a right-infinite word over $\Sigma$.  A {\it factor\/} of $\bf w$ is a finite block of contiguous symbols occurring within $\bf w$.  We let ${\rm Fac}({\bf w})$ denote the collection of factors of ${\bf w}$ (including the empty word).  

We say that two words $v,v'$ over $\Sigma^*$ are \emph{cyclically equivalent}, which we write $v\sim_C v'$, if $v$ and $v'$ are cyclic permutations of one another.  We then let $[v]_C$ denote the equivalence class of $v$ under $\sim_C$.   For example, the equivalence class of the English word {\tt tea} is
$\{ {\tt tea, eat, ate} \}$.  We define the \emph{Lie complexity} of ${\bf w}$ to be the function
\begin{equation}
L_{\bf w}(n):=\#\{[v]_C \colon |v|=n
\text{ and } [v]_C\subseteq {\rm Fac}({\bf w})\}.
\end{equation}
That is, $L_{\bf w}(n)$ counts the number of cyclic equivalence classes of length $n$ with the property that {\it every word\/} in the equivalence class is a factor of ${\bf w}$. This can be contrasted with the cyclic complexity function of Cassaigne, Fici, Sciortino, and Zamboni \cite{CFSZ}, defined
as follows:
$$c_{\bf w}(n):=\#\{[v]_C \colon |v|=n
\text{ and } [v]_C \cap {\rm Fac}({\bf w}) \not= \emptyset \},$$
that is, where ``every word'' in our definition of Lie complexity is replaced by ``some word''.
Observe, in particular from our definition that we have the inequality
\begin{equation}
L_{\bf w}(n)\le c_{\bf w}(n)~{\rm for~}n\ge 1.
\end{equation}
Similarly, we have
\begin{equation}
L_{\bf w}(n)\le a_{\bf w}(n)~{\rm for~}n\ge 1,
\end{equation}
where $a_{\bf w}:\mathbb{N}\to \mathbb{N}$ is the \emph{abelian complexity} function, which counts factors of length $n$ up to abelian equivalence, where $v$ and $v'$ are abelian equivalent if $v'$ can be obtained from $v$ via some permutation of the letters \cite{RSZ}.  

We now show the relation between factors $y$ of ${\bf w}$ of unbounded exponent in ${\bf w}$ and the Lie complexity function.  To make this precise, we construct an equivalence relation $\sim$ on the collection of right-infinite words over $\Sigma$ in which two right-infinite words are equivalent if they have the same set of (finite) factors.  We then let ${\rm Per}({\bf w})$ denote the set of $\sim$ equivalence classes of right-infinite words of the form $v^{\omega}$ such that ${\rm Fac}(v^{\omega})\subseteq {\rm Fac}({\bf w})$.  
The following result is the key estimate, which will be used in proving Theorem \ref{thm:main2}.

 \begin{proposition} Let $\Sigma$ be a finite alphabet and let ${\bf w}$ be a right-infinite word over $\Sigma$.  Suppose that there is a positive number $\kappa$ such that for each positive integer $b\ge 1$, there is a positive integer $n=n(b)$ such that $L_{\bf w}(bn)\le \kappa$.  Then $\#{\rm Per}({\bf w})\le \kappa$.  In particular, if $L_{\bf w}(n)$ is uniformly bounded then $\#{\rm Per}({\bf w})$ is finite.
 \label{prop:kappa}
 \end{proposition}
 \begin{proof}
 Suppose that there exist distinct equivalence classes $[u_1^{\omega}],\ldots ,[u_s^{\omega}]$ in ${\rm Per}({\bf w})$ with $s>\kappa$. Pick $D$ such that $u_i^D$ is not a factor of $u_j^{\omega}$ whenever $i\neq j$.  Let $$b:=D|u_1|\cdot |u_2|\cdots |u_s|.$$  Then by construction, for each $n\ge 1$, the words $u_1^{nb/|u_1|},\ldots ,u_s^{nb/|u_s|}$ are cyclically inequivalent words of length $nb$ with the property that every cyclic permutation occurs as a factor of ${\bf w}$.  Hence $L_{\bf w}(bn) \ge s>\kappa$ for every $n\ge 1$, which contradicts the hypothesis that $L_{\bf w}(bn)$ must be at most $\kappa$ for some positive integer $n$. The result follows.  
 \end{proof}
 

 \section{Algebraic interpretation of Lie complexity}\label{sec:alg}
  We now give a purely algebraic interpretation of the Lie complexity function, which will be used later in proving Theorem \ref{thm:main1}.  To do this, we introduce the \emph{factor algebra} of a right-infinite word ${\bf w}$.  
  
  Let $\Sigma$ be a finite alphabet and let ${\bf w}$ be a right-infinite word over $\Sigma$.  Given a field $k$, we can construct the \emph{factor} $k$-algebra of ${\bf w}$, which we denote by $A_{\bf w}$.  As a vector space, this is just all finite formal $k$-linear combinations of elements of ${\rm Fac}({\bf w})$; that is,
  $$A_{\bf w} = \left\{ \sum_{v\in {\rm Fac}({\bf w})} \lambda_v v\colon \lambda_v\in k, \lambda_v=0~{\rm for~all~but~finitely~many~}v\in {\rm Fac}({\bf w})\right\},$$
  with multiplication of $v,v'\in {\rm Fac}({\bf w})$ defined by declaring that $v\cdot v'$ is the concatenation of $v$ and $v'$ if $vv'$ is again a factor of ${\bf w}$ and $v\cdot v'$ is zero otherwise.  We can then extend the multiplication to general elements of $A_{\bf w}$ by linearity, and so
  $$\left( \sum_{v\in {\rm Fac}(w)} \alpha_v v\right) \left( \sum_{u\in {\rm Fac}(w)} \beta_u u\right) = \sum_{y\in {\rm Fac}(w)} \sum_{\{(u,v)\colon uv=y\}} \alpha_u \beta_v y.$$
  We now introduce some notation that we will use in obtaining Theorem \ref{thm:main1}.
  
  \begin{notation} We make the following assumptions and introduce the following notation.
  \begin{enumerate}
  \item We let $\Sigma=\{x_1,\ldots ,x_d\}$ be a finite alphabet and we let ${\bf w}$ be a right-infinite word over $\Sigma$.
\item  We let $A_{\bf w}$ be the factor algebra of ${\bf w}$ with base field $k=\mathbb{Q}$.
\item We let $V_n$ denote the subspace of $A_{\bf w}$ spanned by the images of factors of ${\bf w}$ of length $n$.
\item We let $W_n$ denote the subspace of $V_n$ spanned by elements of the form $ab-ba$, where $a,b\in {\rm Fac}({\bf w})$ with $|a|+|b|=n$.
\end{enumerate}
\label{notn}
\end{notation}
Notice that since $V_n$ has a basis given by factors of ${\bf w}$ of length $n$, we have
 \begin{equation}
 p_{\bf w}(n) = {\rm dim}(V_n),
 \label{eq:pwV}
 \end{equation}
 where we are taking the dimension as a $\mathbb{Q}$-vector space.
 
 One important remark is that if we adopt the notation from Notation \ref{notn} and we let $x=x_1+\cdots +x_d\in V_1$, then 
 $x^n$ is the sum of all $n$-fold concatenations of $x_1,\ldots ,x_d$.  Each such concatenation will either be $0$ in the factor algebra or will be equal to a factor of ${\bf w}$ of length $n$; moreover, each factor of ${\bf w}$ of length $n$ can be realized as a unique concatenation of length $n$ of these elements.  Thus, when we work in the factor algebra, we have the formula
 \begin{equation}
     (x_1+\cdots +x_d)^n = \sum_{\{v\in {\rm Fac}(w)\colon |v|=n\}} v.\label{eq:xpower}
     \end{equation}

 \begin{lemma} Adopt the assumptions and notation from Notation \ref{notn}.  Then $$L_{\bf w}(n)={\rm dim}(V_n)-{\rm dim}(W_n).$$
 \label{lem:codim}
  \end{lemma}
  
\begin{proof} We fix $n$ and let $m$ denote the dimension of the quotient space $V_n/W_n$.  Let $u_1+W_n,\ldots ,u_m+W_n$ be a basis for $V_n/W_n$ consisting of $W_n$-cosets of factors $u_1,\ldots ,u_d$ of ${\bf w}$ of length $n$.  Observe that every cyclic permutation of $u_i$ must be a factor of ${\bf w}$, since otherwise we could find words $a$ and $b$ such that $u_i=ab$ and such that $ba$ is not a factor of ${\bf w}$.  But this would give that $ba=0$ and $ab=u_i$ in $A_{\bf w}$ and so we would have $ab-ba=u_i$, which would mean that $u_i\in W_n$, which is a contradiction, since $u_i+W_n$ is part of a basis for $V_n/W_n$.

 Furthermore, the $u_i$ must be cyclically inequivalent, since if there were $i$ and $j$ with $i\neq j$ such that some $u_j$ were a cyclic permutation of $u_i$, we could again write $u_i=ab$ and $u_j=ba$ and we would have $u_i-u_j = ab-ba \in W_n$, which again would contradict the independence of $u_1,\ldots ,u_m$ mod $W_n$.
 
Thus $u_1,\ldots ,u_m$ are cyclically inequivalent words such that $[u_1]_C,\ldots ,[u_m]_C$ are all contained in ${\rm Fac}({\bf w})$ and so $$L_{\bf w}(n)\ge m= {\rm dim}(V_n/W_n).$$  
 
 Now we show that $L_{\bf w}(n)\le  {\rm dim}(V_n/W_n)$. Observe that if $L_{\bf w}(n)$ is strictly greater than
 $ {\rm dim}(V_n/W_n)$, then there must exist some word $u_{m+1}\in {\rm Fac}({\bf w})$ of length $n$ such that every cyclic permutation of $u_{m+1}$ is also a factor of ${\bf w}$ and such that $u_{m+1}$ is not cyclically equivalent to $u_i$ for $i=1,\ldots ,m$.  
 
 Since $u_1+W_n,\ldots  ,u_m+W_n$ form a basis for $V_n/W_n$.  By assumption, there exist rational constants
 $\alpha_1,\ldots ,\alpha_m$ such that ${u}_{m+1} - \sum_{i=1}^m \alpha_i {u}_i \in W_n$.
 Then by definition of $W_n$ there are words $a_1,\ldots ,a_s, b_1,\ldots ,b_s$ and rational constants $\beta_1,\ldots ,\beta_s$ such that
\begin{equation}
u_{m+1} - \sum_{i=1}^m \alpha_i u_i = \sum_{i=1}^s \beta_i (a_ib_i-{b_ia_i})
\label{eq:u}
\end{equation}
in the factor algebra $A_{\bf w}$. We let $U$ denote the subspace of $V_n$ spanned by images of words that are cyclically equivalent to $u_{m+1}$ and we define a linear map $\pi: V_n\to U$.  Since $V_n$ has a basis consisting of factors of ${\bf w}$ of length $n$, it suffices to define $\pi$ on such factors and then extend linearly. For a factor $u$ of ${\bf w}$ of length $n$, we define $\pi(u)=u$ if $u\sim_C u_{m+1}$ and $\pi(u)=0$ otherwise.
  
Then since $u_1,\ldots ,u_m, u_{m+1}$ are pairwise cyclically inequivalent, the left side of Equation (\ref{eq:u}) is sent to $u_{m+1}$ by the map $\pi$; the right side, however, is sent to an element of $W_n$, since for each $i$, either $a_ib_i$ and $b_ia_i$ are both cyclically equivalent to $u_{m+1}$ or neither $a_ib_i$ nor $b_ia_i$ is cyclically equivalent to $u_{m+1}$. It follows that $u_{m+1}\in W_n$.  
 
 Thus $u_{m+1}$ is a $\mathbb{Q}$-linear combination of elements of the form $ab-ba$ with each $ab$ and $ba$ cyclic permutations of $u_{m+1}$.  But by assumption, each cyclic permutation of $u_{m+1}$ is in ${\rm Fac}({\bf w})$ and so if we let $T:U\to \mathbb{Q}$ be the linear map uniquely defined by sending $u$ to $1$ for each cyclic permutation $u$ of $u_{m+1}$, we see that $T\circ \pi$ sends the right side of Equation \ref{eq:u} to zero and 
 $T\circ \pi(u_{m+1})=1$, a contradiction.  Thus we obtain the reverse inequality and so $$L_{\bf w}(n)={\rm dim}(V_n/W_n)={\rm dim}(V_n)-{\rm dim}(W_n).$$  
 \end{proof}
 \begin{remark} Notice that $W_n$ is spanned by commutators, $ab-ba$, and that the algebra $A_{\bf w}$ becomes a Lie algebra when endowed with the bracket $[a,b]:=ab-ba$.  It is this fact and Lemma \ref{lem:codim}, which motivates the name Lie complexity for the function $L_{\bf w}$.
 \end{remark}
 \section{Proof of Theorems \ref{thm:main1} and \ref{thm:main2}}\label{sec:proof}
  We can now use the algebraic framework from the preceding section to prove Theorem \ref{thm:main1}.  Our proof adapts an argument with Lie brackets from \cite{B}.
\begin{proof}[Proof of Theorem \ref{thm:main1}]
Let $x=\sum_{i=1}^d x_i\in V_1$.  
 We have a linear map 
 $$\Phi_n : V_n \to W_{n+1}$$ defined by
$$u \mapsto ux-xu = \sum_{i=1}^d ux_i-x_i u$$ for $u\in V_n$.  Then by construction, $\Phi_n$ sends a factor of ${\bf w}$ of length $n$ into the space $W_{n+1}$, and so $\Phi_n$ does indeed map into $W_{n+1}$.    

Recall from Equation (\ref{eq:xpower}) that 
$$x^n=\sum_{\{u \in {\rm Fac}({\bf w})\colon |u|=n\}} {u}.$$
We claim that the kernel of $\Phi_n$ is spanned by $x^n$.
To see this, observe that $\Phi_n(x^n) = x^n\cdot x - x\cdot x^n=0$ and so $x^n$ is in the kernel of $\Phi_n$.

Suppose that there is $$z:=\sum_{\{u \in {\rm Fac}({\bf w})\colon |u|=n\}} \alpha_u u\in {\rm ker}(\Phi_n)$$ with $z$ not in the span of $x^n$. Then we can replace $z$ by $z-\alpha x^n$ for some $\alpha$ in $\mathbb{Q}$ and assume that there is some factor $u$ of ${\bf w}$ of length $n$ with $\alpha_u=0$ but that $z$ is nonzero.  Since $z$ is nonzero, there is some factor $v$ of ${\bf w}$ of length $n$ such that $\alpha_v\neq 0$. 

Then since $v$ and $u$ are both factors of $w$ there is some factor $y_{u,v}$ of ${\bf w}$ that either has $v$ as a prefix and $u$ as a suffix or has $u$ as a suffix and has $v$ as a prefix.  Among all factors $y$ of ${\bf w}$ having the property that $u$ is either a prefix or suffix and having the property that some word $v'$ of length $n$ with $\alpha_{v'}\neq 0$ is either a prefix or a suffix, we pick one, $y_0$, of shortest length possible.  

By symmetry, it suffices to consider the case when $u$ is a prefix of $y_0$, and we let $v'$, with $\alpha_{v'}\neq 0$, denote the suffix of $y_0$ of length $n$. Then 
\begin{equation}
    \label{eq:y0}
y_0=ua=bv'
\end{equation}
for some words $a$ and $b$.  Since $\alpha_u=0\neq \alpha_{v'}$ we see that $|a|=|b|\ge 1$.  

Let $j$ be such that $x_j$ is the last letter of $b$ and write $b=b' x_j$.  
By assumption $\Phi_n(z)=0$ and so
\begin{equation}
    \label{eq:Phiz}
\sum_{\{s \in {\rm Fac}({\bf w})\colon |s|=n\}} \sum_{i=1}^d  \alpha_s(sx_i-{x_i s}) = 0.
\end{equation}
We now consider the coefficient of ${x_jv'}$ in both sides of Equation (\ref{eq:Phiz}).  The coefficient in the right side is equal to zero.  On the other hand, we have 
\begin{equation}
    \label{eq:jkv}
x_j v' = v'' x_k
\end{equation}
for some $k\in \{1,\ldots ,d\}$ and some word $v''$ of length $n$, and so the coefficient of ${x_jv'}$ in the left side of Equation (\ref{eq:Phiz}) is 
$-\alpha_{v'}+ \alpha_{v''}$, since the only contribution from the terms
$\sum_{i=1}^d \alpha_s sx_i$ occurs when $i=k$ and $s=v''$ and the only contribution from the terms $\sum_{i=1}^d -\alpha_s x_i s$ comes when $i=j$ and $s=v'$.

Hence $0=-\alpha_{v'}+\alpha_{v''}$ and so in particular $\alpha_{v''}$ is nonzero.  Then from Equations (\ref{eq:y0}) and (\ref{eq:jkv}) and the fact that $b=b'x_j$, we see
$$ua=bv'=b'x_jv' = b'v'' x_k.$$  Thus $x_k$ is the last letter of $a$ and so $a=a' x_k$ for some word $a'$ with $|a'|<|a|$.  

But now
$$y_0':=ua' = b' v''$$ has the property that $u$ is a prefix, $v''$ is a suffix and $\alpha_{v''}\neq 0$ and $|y_0'|<|y_0|$, which contradicts the minimality of $|y_0|$.  It follows that the kernel of $\Phi_n$ is spanned by $x^n$, and since $x^n$ is nonzero in the factor algebra, the kernel of $\Phi_n$ is exactly one-dimensional.

Then the rank-plus-nullity theorem for linear maps gives that 
 $${\rm dim}(V_n) = {\rm dim}({\rm ker}(\Phi_n))+{\rm dim}({\rm Im}(\Phi_n)).$$
Since we have now established that ${\rm dim}({\rm ker}(\Phi_n))=1$ and since the image of $\Phi_n$ is a subspace of $W_{n+1}$, we have in fact that
 $${\rm dim}(V_n) \le {\rm dim}(W_{n+1}) + 1,$$ or, equivalently,
 $${\rm dim}(W_{n+1})\ge {\rm dim}(V_n)-1.$$
 
Consequently, Lemma \ref{lem:codim} gives
 $$L_{\bf w}(n+1) = {\rm dim}(V_{n+1})-{\rm dim}(W_{n+1}) \le {\rm dim}(V_{n+1}) - ({\rm dim}(V_n)-1).$$ Using Equation (\ref{eq:pwV}), we then see that
 $$L_{\bf w}(n+1)\le p_{\bf w}(n+1)-p_{\bf w}(n) +1,$$ and so we obtain the desired inequality.  When ${\bf w}$ has linear factor complexity, a deep result of Cassaigne \cite{Cass} shows that $p_{\bf w}(n+1)-p_{\bf w}(n)$ is uniformly bounded above by a constant, which then gives that $L_{\bf w}(n)$ is similarly bounded, and so the proof is complete.
 \end{proof}
 
We now get the proof of Theorem \ref{thm:main2}.

 \begin{proof}[Proof of Theorem \ref{thm:main2}] By Theorem \ref{thm:main1}, $L_{\bf w}(n)$ is uniformly bounded when ${\bf w}$ has linear factor complexity. Proposition \ref{prop:kappa} then gives that $\#{\rm Per}({\bf w})$ is finite.  Since there are only finitely many primitive words $y'$ such that $(y')^{\omega}$ has the same set of factors of a fixed periodic right-infinite word, we then obtain the desired result.
  \end{proof}

  \section{Construction}\label{sec:con}
  In this section, we give a construction that proves Theorem \ref{thm:main3}.  We note that we make use of a similar construction given by the first author and Smoktunowicz \cite{BS} in the context of monomial algebras, which we sharpen slightly.
  
  \begin{proof}
  Let $f:\mathbb{N}\to \mathbb{N}$ be a function that tends to infinity.  We shall construct a recurrent binary word ${\bf w}$ whose factor complexity function is bounded above by $n f(n)$ for $n$ sufficiently large such that ${\rm Per}({\bf w})$ is infinite.

  First observe that by replacing $f(n)$ by $\min(f(j)\colon j\ge n)$, we may assume that $f$ is weakly increasing.  Then for each $j$ there is some largest natural number $m_j$ such that $f(m_j)\le 19j^2$.  Then since $f(n)$ is weakly increasing and tends to infinity, we see that the $m_j$ are weakly increasing and tend to infinity. 
  
  To begin, we let ${\bf f}$ be the Fibonacci word, which is a Sturmian word, and hence has complexity function $p_{\bf f}(n)=n+1$.  We fix a prefix $u_1$ of ${\bf f}$ of length $2^{m_1}$.  Since ${\bf f}$ is uniformly recurrent, there are infinitely many occurrences of $u_1$ in ${\bf f}$, and hence we can find a prefix $u_2$ of ${\bf f}$ of length at least $2^{m_1+m_2}$ such that $u_1$ is a suffix of $u_2$. We let $d_2$ denote the length of $u_2$.  In general, for each $i$ there is a prefix $u_i$ of ${\bf f}$ such that $u_{i-1}$ is a suffix and such that $|u_i|>2^{m_i} |u_{i-1}|$, and we let $d_i$ denote the length of $u_i$.   Then $d_i$ is at least $2^{m_i+\cdots +m_1}$.  We define $a_{i,j}= \lceil |u_i|/|u_j|\rceil$ for $i,j\ge 1$ and define
  \begin{equation}
  v_n =u_nu_{n-1}^{a_{n,n-1}}\cdots u_2^{a_{2,1}}u_1^{a_{n,1}} u_2^{a_{n,2}}\cdots u_{n-1}^{a_{n,n-1}}u_n\end{equation} and we define
  a sequence of words $s_n$ with $s_1=u_1$ and for $n\ge 2$ we take
  \begin{equation}
  s_n=s_{n-1}v_n s_{n-1}v_n.\end{equation}
  Since each $s_i$ is a prefix of $s_{i+1}$, we can define the right-infinite word 
  \begin{equation}
  {\bf w}=\lim_n s_n,
  \end{equation}
  and since every factor of ${\bf w}$ appears in some prefix $s_i$ and since $s_{i+1}=s_i v_{i+1} s_i v_{i+1}$, we see that ${\bf w}$ is recurrent.
  
This construction first appears in work of the first author and Smoktunowicz \cite[\S4]{BS} (but in that paper the authors used $W$ for ${\bf f}$, $W_i$ for the prefixes $u_i$, $V_n$ for the factors $v_n$, $U_n$ for the factors $s_n$, and $U$ for the word ${\bf w}$).
 
   Let $n$ be a natural number that is larger than $|u_1|$. Then there is a unique $d$ such that $|u_d|\le n<|u_{d+1}|$. 
   Since $u_d \ge 2^{m_d+\cdots +m_1}$, we see that $n\ge 2^{m_d}$.
  
  Then we may write ${\bf w}=(s_d v_{d+1} s_d v_{d+1})v_{d+2}(s_d v_{d+1} s_d v_{d+1})v_{d+2}\cdots$ and since $|v_j|>2n$ for $j>d$, a factor of ${\bf w}$ of length $n$ is either:
  \begin{enumerate}
 \item a factor of some word of the form $v_js_dv_k$ with $j,k>d$; or
 \item a factor of length $n$ of some 
  $v_i v_{i+1}$ with $i\ge d+1$ that overlaps with both a suffix of $v_i$ and a prefix of $v_{i+1}$.  
  \end{enumerate}
  Then since $u_{d+1}$ is both a prefix and suffix of $v_i$ for $i\ge d+1$ and since $|u_{d+1}|>n$, we see that every factor of $v_js_dv_k$, with $j,k>d$,
  of length $n$ is either a factor of $u_{d+1} s_d u_{d+1}$ or a factor of $v_j$ for some $j$.  Similarly, a factor of $v_i v_{i+1}$ with $i\ge d+1$ that overlaps with both a suffix of $v_i$ and a prefix of $v_{i+1}$ must be a factor of $u_{d+1}^2$ that overlaps with both copies of $u_{d+1}$.
 
 We now consider these three types of factors in a case-by-case basis.  A factor of $u_{d+1} s_d u_{d+1}$ of length $n$ is either a factor of $u_{d+1}$, or it must overlap with $s_d$. Since $u_{d+1}$ is a factor of a Sturmian word, there are at most $n+1$ distinct factors of $u_{d+1}$ of length $n$; there are at most $n-1+|s_d|$ ways of choosing a factor of $u_{d+1} s_d u_{d+1}$ of length $n$ that overlaps with $s_d$.  Thus we see that there are at most $2n+|s_d|$ factors of $u_{d+1}s_du_{d+1}$ of length $n$.
 
 There are $n-1$ ways of selecting a factor of $u_{d+1}^2$ of length $n$ that overlaps with both copies of $u_{d+1}$.  Thus we see that factors of $u_{d+1}^2$ that overlap with both copies of $u_{d+1}$ contribute at most $n-1$ additional factors to our count. 
 
 Finally, there are at most $12 d^2n$ factors of some $v_j$ \cite[Lemma 4.4]{BS} and so we see that the total number of factors of ${\bf w}$ of length $n$ is at most
  $$3n-1+|s_d|+12 d^2n.$$
  
 Now \cite[Equation (4.8)]{BS} gives that $|s_d|\le 4d^2|u_d| \le 4d^2n$ and so the number of factors of ${\bf w}$ of length $n$ is at most
  $3n-1+4d^2n + 12 d^2 n \le 19n d^2$.  We have $n\ge 2^{m_d}$ and since $f(j)>19d^2$ for $j>m_d$, we see that $19n d^2 \le n f(n)$, and so $p_{\bf w}(n)\le n f(n)$ for $n\ge |u_1|$, which gives the desired bound on the factor complexity of ${\bf w}$.
  
  Finally, observe that for a fixed $i$, the word $u_i^{a_{n,i}}$ appears as a factor of $v_{n}$ and hence as a factor of ${\bf w}$.  Since $a_{n,i}\ge |u_n|/|u_i| \to \infty$, we see that arbitrarily large powers of $u_i$ appear as factors of ${\bf w}$.   Now for each $i$, there is some primitive word $y_i$ such that $u_i=y_i^{e_i}$ for some $e_i\ge 1$.  Since $u_i$ is a factor of the Fibonacci word and since the Fibonacci word is $4$th-power free \cite{MR}, we see that $e_i\in \{1,2,3\}$ for every $i$.  Hence $|y_i|\to \infty$ as $i\to\infty$, and so we have infinitely many primitive words $y$ such that $y^n$ is a factor of ${\bf w}$ for every $n$.
\end{proof}
  
\section{Automatic sequences}\label{sec:aut}

A sequence ${\bf s} = (s_n)_{n \geq 0}$ is {\it $k$-automatic\/} if
there exists a finite automaton
that, on input the base-$k$ representation of $n$, computes $s_n$ (by arriving at a state whose output is $s_n$).
We have the following result \cite{Charlier&Rampersad&Shallit:2012}:
\begin{theorem}
\label{thm1}
Let $\bf s$ be a $k$-automatic sequence.  
\begin{itemize}
\item[(a)] There is an algorithm that,
given a well-formed first-order logical formula $\varphi$ in
$\theory(\N, +, 0, 1, n \rightarrow {\bf s}[n])$
having no free variables, decides if $\varphi$ is true or false.
\item[(b)]
Furthermore, if $\varphi$ has free variables, then the algorithm
constructs an automaton recognizing the representation of the values
of those variables for which $\varphi$ evaluates to true.
\end{itemize}
\end{theorem}

A sequence $(a_n)_{n \geq 0}$ taking values in $\Z$ is {\it $k$-regular\/} if
there is a linear representation for it, that is, a row vector $v$, a column vector $w$, and a matrix-valued morphism $\zeta:\{0,1,\ldots, k-1\} \rightarrow \Z^{d\times d}$ such that $a_n = v \cdot \zeta(x) \cdot w$,
where $x$ is the base-$k$ representation of $n$.
If $A$ is an automaton accepting the base-$k$ representation of pairs $(i,n)$ in parallel, then the sequence
$a_n = \# \{ i: A \text{ accepts } (i,n) \}  $ is $k$-regular, and furthermore the matrices $\zeta(a)$ in the linear representation for $(a_n)$ have
non-negative integer entries \cite{Charlier&Rampersad&Shallit:2012}.
 A $k$-regular sequence taking only finitely many values is $k$-automatic \cite[Thm.~16.1.5]{Allouche&Shallit:2003},
 and the automaton can be algorithmically produced from the linear representation because the entries of $\zeta(a)$ are in
 $\N$.

In this section we prove Theorem~\ref{thm:main5}: if $\bf w$ is a $k$-automatic sequence, then the sequence $(L_{\bf w} (n))_{n \geq 0}$ is also $k$-automatic.

\begin{proof}
We will show that the sequence $(L_{\bf w} (n))_{n \geq 0}$ is $k$-regular.  Since automatic sequences have linear factor complexity \cite[Thm.~10.3.1]{Allouche&Shallit:2003}, it follows from Theorem~\ref{thm:main2}
that $(L_{\bf w} (n))_{n \geq 0}$ is bounded, and hence  automatic.

We construct a linear representation for $(L_{\bf w} (n))_{n \geq 0}$ by constructing a first-order logical formula $\lie(i,n)$ for
the pairs $(i,n)$ such that
\begin{itemize}
    \item[(a)] All of the cyclic shifts of 
${\bf w}[i..i+n-1]$  appear in $\bf w$;
\item[(b)] ${\bf w}[i..i+n-1]$ is the lexicographically least of all its cyclic shifts appearing in $\bf w$; and
\item[(c)] ${\bf w}[i..i+n-1]$ is the first
occurrence of this particular factor.
\end{itemize}
Then the number of $i$ making $\lie(i,n)$ true equals
$L_{\bf w} (n)$.

We do this in a number of steps:
\begin{itemize}

\item $\factoreq(i,j,n)$ asserts that 
the length-$n$ factor ${\bf w}[i..i+n-1]$ equals 
${\bf w}[j..j+n-1]$;

\item $\shift(i,j,n,t)$ asserts that
${\bf w}[i..i+n-1]$ is the shift, by $t$ positions, of the factor
${\bf w}[j..j+n-1]$.

\item $\conj(i,j,n)$ asserts that the factor
${\bf w}[i..i+n-1]$ is a cyclic shift of
${\bf w}[j..j+n-1]$.

\item $\lessthan(i,j,n)$ asserts that the
factor ${\bf w}[i..i+n-1]$ is lexicographically smaller than
${\bf w}[j..j+n-1]$.

\item $\lessthaneq(i,j,n)$ asserts that the
factor ${\bf w}[i..i+n-1]$ is
lexicographically $\leq$ the factor
${\bf w}[j..j+n-1]$.

\item $\allconj(i,n)$ asserts that all cyclic shifts ${\bf w}[i..i+n-1]$ appear
as factors of ${\bf w}$.

\item $\lexleast(i,n)$ asserts that ${\bf w}[i..i+n-1]$ is 
lexicographically least among all its cyclic shifts that actually appear in $\bf w$.

\item $\lie(i,n)$ asserts that all cyclic shifts of ${\bf w}[i..i+n-1]$
appear in $\bf w$, that ${\bf w}[i..i+n-1]$ is the
lexicographically least cyclic shift, and that
${\bf w}[i..i+n-1]$ is its first occurrence in $\bf w$.

\end{itemize}

Here are the definitions of the formulas.   Recall that the domain of all variables is $\N = \{0,1,\ldots \}$.
\begin{align*}
\factoreq(i,j,n) &:= \forall u, v \, (i+v=j+u \, \wedge \, u\geq i \, \wedge\, u<i+n) \implies {\bf w}[u] = {\bf w}[v] \\
\shift(i,j,n,t) &:= \factoreq(j, i+t,n-t) \, \wedge \, \factoreq(i, (j+n)-t, t) \\
\conj(i,j,n) & := \exists t\, (t\leq n) \, \wedge\, \shift(i,j,n,t) \\
\lessthan(i,j,n) & := \exists t \, (t<n) \, \wedge\, \factoreq(i,j,t) \, \wedge\, {\bf w}[i+t]<{\bf w}[j+t] \\
\lessthaneq(i,j,n) &:= \lessthan(i,j,n) \, \vee \, \factoreq(i,j,n) \\
\allconj(i,n) & := \forall t \, (t\leq n) \implies \exists j\, \shift(i,j,n,t) \\
\lexleast(i,n) & := \forall j \, \conj(i,j,n) \implies \lessthaneq(i,j,n) \\
\lie(i,n) & := \allconj(i,n) \, \wedge\, \lexleast(i,n) \, \wedge\, (\forall j \, \factoreq(i,j,n) \implies (j \geq i)) 
\end{align*}
From the remarks preceding the proof, we are now done.
\end{proof}

\begin{remark}
Most of the logical formulas should be self-explanatory, with one exception:  in order to specify 
$\allconj$, why do we use shifts of length $0, 1, \ldots, n$?  It is because we want the formula to work even in the case of the empty word.
\end{remark}

\begin{corollary}
Given an automatic sequence $\bf w$, the quantity
$\sup_{n \geq 0} L_{\bf w} (n) $ is computable.
\label{cor5}
\end{corollary}

\begin{remark}
Theorem~\ref{thm:main5} and Corollary~\ref{cor5} also hold for automata based on other kinds of numeration systems, such as
Fibonacci numeration \cite{Frougny:1986}; Tribonacci numeration
\cite{Mousavi&Shallit:2015}; and Ostrowski numeration systems
\cite{Baranwal:2020}.
\end{remark}

\section{Examples}\label{sec:exam}

Using the free software {\tt Walnut} \cite{Mousavi:2016}, we can implement the algorithm of the previous section to find automata and closed-form expressions for
$L_{\bf w} (n)$ for some classical words of interest.

\begin{example} Let $\bf t$ be the Thue-Morse word, the fixed point of the morphism $\mu$ sending $0$ to $01$ and $1$ to $10$.
Then 
$$ L_{\bf t} (n) = \begin{cases}
1, & \text{if $n = 0$ or $n = 2^k$ for $k \geq 3$}; \\
2, & \text{if $n = 1,4$ or $n = 3 \cdot 2^k$ for $k \geq 0$}; \\
3, & \text{if $n = 2$}; \\
0, & \text{otherwise.}
\end{cases} $$
\end{example}
To some extent this is not surprising, since we know that the only squares in $\bf t$ are of length $2^k$ or $3 \cdot 2^k$.  However, $L_{\bf w} (n)$ can be nonzero even if a sequence has no squares, as the following example shows.
\begin{example} Let $\bf vtm$ be the variant of the Thue-Morse word defined over a ternary alphabet,
the fixed point of the morphism sending $2$ to $210$, $1$ to $20$, and $0$ to $1$.   It is well-known that $\bf vtm$ is squarefree \cite{Berstel:1978}.
Then 
$$ L_{\bf vtm} (n) = \begin{cases}
1, & \text{if $n = 0$ or $n = 2^k$ for $k \geq 2$}; \\
2, & \text{if $n = 3 \cdot 2^k$ for $k \geq 0$}; \\
3, & \text{if $n = 1,2$}; \\
0, & \text{otherwise.}
\end{cases} $$
\end{example}

\begin{example}  Let us look at an example in a different base, and where there are factors of unbounded exponent.
Let ${\bf c} = 101000101 \cdots$ be the {\it Cantor sequence}, which is the fixed point of the morphism $1 \rightarrow 101$
and $0 \rightarrow 000$.  Then
$$ L_{\bf c} (n) = \begin{cases}
3, & \text{if $n = 4$}; \\
2, & \text{if $n = 0,1,3$ or $2 \cdot 3^k$ for $k \geq 0$}; \\
1, & \text{otherwise.}
\end{cases} $$
\end{example}

\begin{example} Let $\bf f$ be the Fibonacci word,
the fixed point of the morphism sending $0$ to $01$ and $1$ to $0$.  Define the Fibonacci numbers by
$F_0 = 0$, $F_1 = 1$, and $F_n = F_{n-1} + F_{n-2}$ for $n \geq 2$.
Then 
$$ L_{\bf f} (n) = \begin{cases}
1, & \text{if $n = 0$ or $n = F_k$ for $k \geq 4$ or $n = F_k + F_{k-3}$ for $k \geq 4$ }; \\
2, & \text{if $n = 1,2$}; \\
0, & \text{otherwise.}
\end{cases} $$
\end{example}

\begin{example} Let $\bf TR$ be the Tribonacci word,
the fixed point of the morphism sending $0$ to $01$,
$1$ to $02$, and $2$ to $0$.  Define the Tribonacci numbers by
$T_0 = 0$, $T_1 = 1$, $T_2 = 1$, and $T_n = T_{n-1} + T_{n-2} + T_{n-3}$ for $n \geq 3$.
Then 
$$ L_{\bf TR} (n) = \begin{cases}
1, & \text{if $n = 0$ or $n = T_k$ for $k \geq 5$ or $n = T_k + T_{k-1}$ for $k \geq 3$ or} \\
& \text{$n = T_k + T_{k-4}$ for $k \geq 5$ }; \\
2, & \text{if $n = 4$}; \\
3, & \text{if $n = 1,2$}; \\
0, & \text{otherwise.}
\end{cases} $$
\end{example}

Finally, we give an example where $L_{\bf w} (n) = 0$
for $n \geq 2$:

\begin{example}
Let $\Sigma=\{x_1,\ldots ,x_6,y_1,\ldots ,y_6\}$
and let $\Phi:\Sigma^*\to \Sigma^*$ be the morphism given by
\[\begin{array}{lll} x_1\mapsto  x_1x_2y_1y_2 &  x_2\mapsto x_1x_3y_1y_3 & x_3\mapsto x_1x_4y_1y_4 \\
 x_4\mapsto x_1x_5y_1y_5 & x_5\mapsto x_1x_6y_1y_6 & x_6\mapsto x_2x_3y_2y_3 \\
 y_1\mapsto x_2x_4y_2y_5 & y_2\mapsto x_2x_5y_3y_4 & y_3\mapsto x_2x_6y_2y_6\\
  y_4\mapsto x_3x_4y_3y_5  &y_5\mapsto x_3x_5y_3y_6 & y_6\mapsto x_3x_6y_4y_5.
  \end{array}\]  and let $ {\bf w}=\Phi^{\omega}(x_1)$.  Then $\bf w$ is $2$-automatic and Lemma 6.1 of \cite{BM} shows that $L_w(n)=0$ for $n\ge 2$.
  \end{example}
The preceding example gives a word in which every factor of length at least two has some cyclic conjugate that is not a factor.  Badkobeh and Ochem \cite{Bad} give an example of such a word over a $5$-letter alphabet. In general, the property that $L_{\bf w}(n)=0$ for $n\ge i$ has been studied over various alphabets \cite{Gam}.
  
  \section{Concluding Remarks}
  \label{sec:Q}
In this final section, we pose a question that is suggested by the computations we've performed.
\begin{question} Can the Lie complexity function of a morphic word be unbounded?
\label{Q1}
\end{question}
We note that the analogue of Theorem \ref{thm:main2} is known to hold for pure morphic words \cite[Corollary 20]{KS}, and so if Question \ref{Q1} has an affirmative answer, this would give an extension of this result to general morphic words.

Although many classes of morphic words, including primitive morphic and $k$-uniform morphic words, have linear complexity and hence are covered by Theorem \ref{thm:main2}, the factor complexity function of a morphic word need not be linear in general.  Pansiot (see \cite[Theorem 4.7.1]{BR}) has shown that the factor complexity of a pure morphic word is either ${\rm O}(1)$, $\Theta(n)$, $\Theta(n \log \log\, n)$, $\Theta(n \log\, n)$, or $\Theta(n^2)$, and that each of these possibilities can be realized as the factor complexity of a pure morphic word.

\end{document}